\documentclass[amsthm]{elsart}
\usepackage{graphicx}
\usepackage[mathscr]{eucal}
\usepackage{amsfonts}
\usepackage{amsmath,amssymb,amscd,amsthm}
\usepackage[usenames,dvipsnames]{pstricks}

\newtheorem{theorem}{Theorem}
\newtheorem{lemma}{Lemma}

\newtheorem{proposition}{Proposition}
\theoremstyle{definition}
\newtheorem{definition}{Definition}

\input{epsf}
\newcommand{\beq}{\begin{equation}}

\newcommand{\M}{{\bf M}}
\newcommand{\B}{{\bf B}}
\newcommand{\N}{{\bf N}}
\newcommand{\sgn}{\mbox{sgn}}

\newcommand{\eeq}{\end{equation}}
\newcommand{\R}{{\mathbb R}}
\newcommand{\pd}{\partial}
\begin{document}
\begin{frontmatter}
\title{Block Regularization of the Kepler Problem on Surfaces of Revolution with  Positive Constant Curvature.}

%\maketitle
\author{Manuele Santoprete}
\address{Department of Mathematics\\
Wilfrid Laurier University\\
75 University Avenue West,\\
Waterloo, ON, Canada, N2L 3C5.}
\ead{msantopr@wlu.ca}

%\date{}                                           % Activate to display a given date or no date
\begin{abstract}
We consider the Kepler problem on  surfaces of revolution that are  homeomorphic to $S^2$ and have constant Gaussian curvature. We show that the system is maximally superintegrable, finding constants of motion that generalize the Runge-Lentz vector.  Then, using such first integrals, we determine the class of surfaces  that lead  to block-regularizable collision singularities. In particular we show that the singularities are always regularizable if the surfaces are spherical orbifolds of revolution with constant curvature.   
\end{abstract}

\begin{keyword}
Block Regularization \sep Kepler problem \sep surfaces of constant curvature
%\PACS 
\end{keyword}
\end{frontmatter}
%\markboth{M. Santoprete}{Regularization of the Kepler problem in spaces of constant curvature}

%%%%%%%%%%%%%%%%%%%%%%%%
%\section*{\large\bf I. Introduction}
\section{Introduction}
%%%%%%%%%%%%%%%%%%%%%%%%%

The problem of describing the motion of a particle on surfaces of constant curvature, under the influence of the analogue of the gravitational potential, is an interesting problem that dates back to the 19th century. A short history of the problem is presented in  \cite{Santoprete}.  Many of the classical results have been long forgotten. However, since then, interest in the problem has been revived, at least in part, because of cosmological models as the one based on the  Robertson-Walker  metric. This model   describes a homogeneous, isotropic expanding or contracting universe, and its spatial slices are, depending on the curvature, three-spheres $S^3$, copies of the Euclidean space $\R^3$, or copies of the hyperbolic space $H^3$. 
In recent years many authors have studied the classical Kepler problem and the quantum analogue (the hydrogen atom), rediscovering the old results and introducing new ones (the interested reader can consult \cite{Carinena} for some interesting results and an extensive bibliography) 

In this paper we study the Kepler problem on surfaces  of revolution of constant Gaussian curvature with certain type of metric singularities. This  problem is more general than most of the previous work that concentrates on constant curvature spaces with no singularities (i.e. the sphere $S^2$ and the hyperbolic plane $H^2$, or in higher dimensions the three sphere $S^3$ and the hyperbolic space $H^3$).

As in the standard  Kepler problem the potential is singular and this introduces singularities in the equations and  in the solutions. 
In our case we will consider spherical surfaces of revolution that have constant curvature, i.e. surfaces of revolution that are homeomorphic to the two-sphere $S^2$
and that can have metric singularities at the ``north'' and ``south'' poles.
 
%From the point of view of dynamical systems a sistem of ordinary differential equations is viewed as a vector field on a manifold and its solutions are viewed as determining a flow on that manifold. Thus one attahces a great importance to the smootheness of the flow with respect to initial data. 
Levi-Civita has given us a method for extending orbits of the Kepler problem through double collisions \cite{Levi-Civita}. He eliminated the singularities in the vector field by transforming the equations to ones without singularities. The extension through double collisions then is given by the transformed equations and is automatically a smooth function of initial data. In \cite{Santander} this method is applied to the case of the Kepler problem on spaces of constant curvature (with no metric singularities).

Another  method of regularization  is the topological regularization due to Easton \cite{Easton}, that is based on the general theory developed by Conley and Easton \cite{ConleyEaston} and it is usually called regularization by surgery, or (following McGehee \cite{McGehee81}) block regularization. This method is used  in this paper to regularize the singularities at the poles. 
 Roughly, the idea is to excise a neighborhood of the singularity (more precisely an isolating block) from the manifold on which the vector field is defined and then to identify appropriate points on the boundary of the region.  
This is done constructing a map across the block that identifies the point where a solution enters the block with the point where  a solutions exits the block. This map is a diffeomorphism. However there are solutions that, once they enter the isolating block, never leave it. If the summentioned map can be extended (in a differentiable way) to consider such solutions then one says that the singularity is block regularizable. 

We characterize, as in \cite{Santoprete}, the spherical surfaces of revolution using a parameter $\beta$. In \cite{Santoprete} we showed that whenever $\beta$ is a rational number the Kepler system has the Bertrand property (i.e. all the bounded non-singular orbits are closed)
Here we find that, while the south pole singularity is always block regularizable, 
only few  values of $\beta$ produce a singularity at the north pole that can be regularized according to Easton. Furthermore, it turns out that the north pole singularity is always block-regularizable   in the case of closed (compact, without boundary) surfaces of revolution (with constant curvature) that are Riemannian 2-orbifolds.

 Loosely speaking, a 2-orbifold is modeled locally by convex Riemannian surfaces modulo finite groups of isometries acting with possible fixed points. This means that a neighborhood of each point $p$ of such orbifold is isometric to a Riemannian quotient $U_p/\Gamma_p$ where $U_p$ is a convex Riemannian surface diffeomorphic to $\R^2$, and $\Gamma_p$ is a finite group of isometries acting effectively on $U_p$. Every Riemannian surface is trivially an orbifold, with each $\Gamma_p$ being the trivial group. The reader interested in more background on orbifolds can consult \cite{Thurston}. For the purpose of this paper, however, we only need to apply a simple explicit criterion to determine whether a closed surface of revolution is a 2-orbifold (see section 10 and \cite{Borzellino}). 

The paper is organized as follows. In Section 2 we introduce the  the generalized Kepler potential and the equations of motion. In Section 3 we introduce Gaussian curvature and some properties of constant curvature surfaces. In the following section we find two additional  integrals of motion, besides the Hamiltonian and the angular momentum. In Section 5 we find the equations of the trajectory. In Section 6 we use a transformation to rewrite the equations in a more convenient way and we define the collision manifold. In the  following section we study the flow on the collision manifold. In Section 8 find  that for only few values of  $\beta$ there  is a trivializable  isolating block about collisons.
In Section 9 we determine for which surfaces the singularities are regularizable. 
In the last section we show that the singularities are always block regularizable if the surface under consideration is an orbifold of revolution.

%%%%%%%%%%%%%%%%%%%%%%%%
%\section*{\large\bf II. Equations of motion}
\section{Equations of Motion}
%%%%%%%%%%%%%%%%%%%%%%%%%

Let $I$ be an interval of real numbers then we say that  $\gamma:I\rightarrow \R^2$ is a regular plane curve if $\gamma$ is $C^1$  and
$\gamma'(x)\neq 0$ for any $x\in I$.
% The curve $\gamma$ is said to be %simple if it injective in the interior of $I$. 
%It is said to be closed if $I$ is a closed bounded interval $[a,b]$ and $\gamma(a)=\gamma(b)$.

\begin{definition}Let $\gamma: I\rightarrow \R^3$ be a simple (no self intersections) regular plane curve $\gamma(r)=(f(r),0,g(r))$ on the $xz$-plane where $f$ and $g$ are smooth curves on the interval $I=[r_N,r_S]$, with $f(r)\geq 0$  and $f(r)=0$ if and only if $r=r_N$ or $r=r_S$. 
A spherical surface of revolution $S$ is a surface isometrically  embedded in $\R^3$ that admits a parametrization ${\bf x}:I\times \R\rightarrow S$ of the form
\beq{\bf x}(r,\theta)=(f(r)\cos\theta,f(r)\sin\theta,g(r))
\label{eqparam}\eeq
That is,  $S$ is the surface of revolution obtained by rotating $\gamma$ about the $z-$axis. The curve $\gamma$ will be called the profile curve. 
\end{definition}

Note that a spherical surface of revolution is homeomorphic to $S^2$ and that by definition the sets ${\bf x }(r_N,\theta)$ and ${\bf x}(r_S,\theta)$ reduce to  single points, i.e. the {\it north} and the {\it south poles} of $S$.  
Metric singularities can only  occur at the north and south poles, S is smooth everywhere else (see figure \ref{surface}).

 %-------------------------------------------------------------------------------------
\begin{figure}[t]
  \begin{center}
      {\resizebox{!}{5.5cm}{\includegraphics{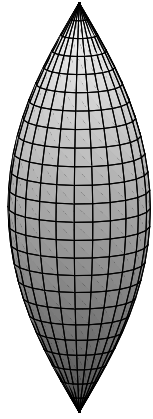}}  }
    \caption{A spherical surface of revolution with positive constant curvature\label{surface}} 
  \end{center}

\end{figure}
%--------------------------------------------------------------------------

Throughout this paper all surfaces of revolution will be assumed to be  as in Definition 1  and the profile curve $\gamma$ is assumed to be unit speed, i.e. $ (\frac{df}{dr})^2+ (\frac{dg}{dr})^2=1$.

For a surface of revolution $S$, a simple computation gives the coefficients of the first fundamental form, or metric tensor (subscripts denote partial derivatives): 
\[E={\bf x}_r \cdot {\bf x}_r=\left(\frac{df}{dr}\right)^2+ \left(\frac{dg}{dr}\right)^2=1, \quad F={\bf x}_r \cdot {\bf x}_\theta=0 \quad G={\bf x}_\theta \cdot {\bf x}_\theta=f(r)^2,\]
 so that the metric (away from any singular point) is 
\beq
ds^2= E~dr^2+2F~dr~d\theta+G~d\theta^2=  dr^2+f(r)^2d\theta^2.
\label{eqds}\eeq
Note that the parametrization is orthogonal ($F=0$) and that $E_\theta=G_\theta=0$.
Surfaces given by parametrizations with these properties are said to be {\it $r$-Clairaut}.
The Lagrangian function of a particle of mass $m$ moving on the surface takes the form
\[
L=\frac m 2 ( \dot r^2+f(r)^2\dot \theta^2) -V(r)
\]
where $V(r)$ is the {\it generalized gravitational potential}, that is  
\beq
V(r)=\gamma \Theta(r)
\label{eqV}
\eeq
where $\gamma$ is a positive constant and $\Theta(r)$ is an antiderivative of $1/f(r)^2$.  The generalized gravitational potential is a solution of  the Laplace-Beltrami equation 
\beq
\triangle V(r)=\frac{1}{f(r)^2}\frac{\partial}{\partial r}\left(f(r)^2\frac{\partial V(r)}{\partial r}\right)=0.
\label{eqlaplace}\eeq
The Hamiltonian is 
\[
H=\frac{p_r^2}{2m}+\frac{p_\theta^2}{2m f(r)^2}+V(r)
\]
where $p_\theta=mf(r)^2\dot\theta$.

The equations of motion are:
\beq\left\{\begin{array}{l}
\dot r=\frac{\partial H}{\partial p_r}=\frac{p_r}{m}\\
\dot\theta=\frac{\partial H}{\partial p_\theta}=\frac{p_\theta}{m f(r)^2}\\
\dot p_r=-\frac{\partial H}{\partial r}=\frac{p_\theta^2f'(r)}{mf(r)^3}    -\frac{dV}{dr} \\
\dot p_\theta=-\frac{\partial H}{\partial \theta}=0
\end{array}\right.
\label{eqmotion}
\eeq
Clearly $H$ and $p_\theta$ are constant of motions, they are in involution and the problem is integrable by  the Liouville-Arnold theorem.

%%%%%%%%%%%%%%%%%%%%%%%%%%%%%
%\section*{\large\bf IV. Gaussian Curvature of Surfaces of Revolution}
\section{Curvature}
%%%%%%%%%%%%%%%%%%%%%%%%%

It is well known that the (Gaussian) curvature of the metric (\ref{eqds})
is given by the equation
\[
K=-\frac{f''(r)}{f(r)}\
\]
Throughout this paper all the metrics will be assumed to be of  constant curvature $K$.
The requirement of constant curvature gives us a linear differential equation to solve
\[
f''=-K f.
\]
The solutions to this differential equation are of the form
\[
f(r)=A_0 e^{i\sqrt{K}r}+B_0e^{-i\sqrt{K}r}
\]
if $K\neq 0$
and 
\[
f(r)=Cr+D
\]
if $K=0$. 

Since we are interested only in spherical surfaces of revolution we can restrict our attentionto the case of positive curvature, i.e. $K>0$.
In this case the solutions take the form
\[
f(r)=A\sin(\sqrt{K}r)+B\cos(\sqrt{K}r)\quad \mbox{or} \quad f(r)=L\sin(\sqrt{K}r+\alpha)
\]
where $A=L\cos\alpha$ and $B=L\sin\alpha$ with $L>0$.

%\begin{remark}
%The sphere is obtained when  $K>0$ and $A=-B=\frac{1}{2i}$.
%The hyperbolic plane is obtained when $K<0$ and $A=-B=\frac 1 2$. 
%\end{remark}

Now we can prove few useful facts and formulas.

\begin{proposition}
The equation 
\beq-ff''+(f')^2=\beta^2
\label{eqf1}\eeq 
is verified if and only if the metric  has constant Gaussian curvature $K$
and either $f(r)=A_0 e^{i\sqrt{K}r}+B_0e^{-i\sqrt{K}r}$ with $A_0B_0=\beta^2/4K$ or 
$f(r)=Cr+D$ with   $C=\pm \beta$. Moreover if the curvature $K$ is positive, then $A^2+B^2=\beta^2/K$.
\label{propf}
\end{proposition}
\begin{proof}
Note that 
\beq
\left(\frac{(f')^2-\beta b^2}{f^2}\right)'=-2\frac{ff'}{f^4}(-ff''+(f')^2-\beta^2)
\label{eqcurv}
\eeq

If  $-ff''+(f')^2=\beta^2$ then from Eq. (\ref{eqcurv}) it follows that
\[
\left(\frac{(f')^2-\beta^2}{f^2}\right)=-K
\]
for some constant $K$. Consequently, since $-ff''+(f')^2=\beta^2$, $f''/f=-K$
and the curvature is constant. 

On the other hand assume that $f''=-Kf$. Then, 
if $K\neq 0$, $f=A_0e^{i\sqrt{K}r}+B_0e^{-i\sqrt{K}r}$. Plugging this into  
$-ff''+(f')^2=\beta^2$ we find the condition $A_0B_0=\frac{\beta^2}{4K}$.
If $K=0$ then $f=Cr+D$. Plugging into the equation we find
$C^2=\beta^2$.

\end{proof}

\begin{proposition}
The function $f$ satisfies the equation \beq
\frac{f'(r)}{f(r)}=-\beta^2\Theta(r)
\label{eqint}
\eeq
for some antiderivative $\Theta(r)$ of $1/f(r)^2$, if and only if it satisfies the  nonlinear differential equation
\[
-ff''+(f')^2=\beta^2
\]\label{propintegral}
\end{proposition}
\begin{proof}
\[
\left(\frac{f'}{f}\right)'=\frac{f''f-(f')^{2}}{f^{2}}=-\frac{\beta^{2}}{f^{2}},
\]
which implies (\ref{eqint}) for some $\Theta(r)$.
\end{proof}
\begin{proposition}
 If $f$ satisfies the equation $-ff''+(f')^2=\beta^2$, and $\Theta(r)$ is defined as above, then
\[
\Theta^2=\frac{1}{f^2\beta^2}-\frac{K}{\beta^4}.
\]
\label{proptheta2}
\end{proposition}
\begin{proof}
 From Proposition \ref{propintegral} and Proposition \ref{propf} we obtain
\[
\Theta^2=\frac{(f')^2}{f^2\beta^4}=\frac{\beta^2+ff''}{f^2\beta^4}.
\]
Since, by Proposition \ref{propf}, the Gaussian curvature is constant we can substitute $f''=-Kf$ in the last equation. The result follows.
\end{proof}

%%%%%%%%%%%%%%%%%
\section{Additional Integrals}
%%%%%%%%%%%%%%%%
In Section 2 we found that the system under consideration has two integrals of motion: the Hamiltonian and the angular momentum.
In this section  we consider the Kepler problem on surfaces of revolution that have constant Gaussian curvature and we find two additional integrals of motion (quadratic in the momenta).  Three of the first integrals are independent and thus the  system is maximally superintegrable.

We now look for first integrals that are quadratic in the momenta. The most general form of an invariant, quadratic in the momenta is

\[
I=ap_r^2+2bp_rp_{\theta}+cp_{\theta}^2+\Phi(r,\theta)
\] 
where $a,b,c$ and $\Phi$ are function of $r$ and $\theta$. 
Given the form of the constant of motion , it is straightforward to find the conditions that must be satisfied to grant its existence. The Poisson bracket of $I$ with the Hamiltonian is
\[\begin{split}
\{I,H\}=&\left(\frac{\pd I}{\pd r}\frac{\pd H}{\pd p_r}-\frac{\pd I}{\pd p_r}\frac{\pd H}{\pd r}\right)+\left(\frac{\pd I}{\pd \theta}\frac{\pd H}{\pd p_{\theta}}-\frac{\pd I}{\pd p_{\theta}}\frac{\pd H}{\pd \theta}\right)\\
=&a_rp_r^3+\left(2b_r+\frac{a_{\theta}}{f(r)^2}\right)p_r^2p_\theta+
\left(c_r+\frac{2af'}{f^3}+\frac{2b_\theta}{f^2}\right)p_rp_\theta^2+\left(\frac{2bf'}{f^3}+\frac{c_\theta}{f^2}\right)p_\theta^3\\
&+(\Phi_r-2aV_r)p_r+\left(-2bV_r+\frac{\Phi_\theta}{f^2}\right)p_\theta
\end{split}
\]
where we used the fact that $V$ is a function of $r$ only and thus $V_\theta=0$.

For $I$ to be a constant of motion the Poisson bracket $\{I,H\}$ must vanish for every value of the momenta. The vanishing of the Poisson bracket  implies the set of equations
\beq\left\{\begin{array}{l}
a_r=0\\
2\frac{bf'}{f^3}+\frac{c_\theta}{f^2}=0\\
2b_r+\frac{a_\theta}{f^2}=0\\
c_r+2\frac{af'}{f^3}+2\frac{b_\theta}{f^2}=0
\end{array}
\right.
\label{eqcond1}\eeq
and
\beq
\Phi_r-2aV_r=0, \quad\quad -2bV_r+\frac{\Phi_\theta}{f^2}=0
\label{eqcond2}\eeq
We now consider the motion under the generalized gravitational potential and we look for additional integrals of motions.
Thus let $V(r)=\gamma\Theta(r)$ then the equations (\ref{eqcond2})
take the form
\beq
\Phi_r=2\frac{\gamma a}{f^2},   \quad\quad \Phi_\theta=2\gamma b.
\eeq
The compatibility condition for the last two equations leads to 
\[
\Phi_{r\theta}=2\frac{\gamma a_\theta}{f^2}=2\gamma b_r=\Phi_{\theta_r}
\]
or $a_\theta=f^2(r)b_r$.
On the other hand from (\ref{eqcond1}) we have $a_\theta=-2b_rf^2(r)$ and thus $b_r=0$ and $b=b(\theta)$.
Moreover by (\ref{eqcond1}) we also have $a=a(\theta)$ and, since $b_r=0$ we have that $a_\theta=0$. Consequently $a=\mbox{constant}$.

Integrating the last equation of system (\ref{eqcond1}) we obtain
\[
c(r,\theta)=\frac{a}{f^2}-2b_\theta\frac{V(r)}{\gamma}+\Gamma(\theta)
\]
where $\Gamma$ is an arbitrary function of $\theta$. Differentiating with respect to $\theta$ yields
\[
c_\theta=-2b_{\theta\theta}\frac{V(r)}{\gamma}+\Gamma_\theta.
\]
On the other hand  the second equation of  system  (\ref{eqcond1})
implies
\[
c_\theta=-2\frac{bf'}{f}.
\]
Therefore we obtain
\[
-2b_{\theta\theta}\frac{V}{\gamma}+\Gamma_\theta=-2b\frac{f'}{f}.
\]
Let $V=\gamma \Theta(r)$ and let $-ff''+(f')^2=\beta^2$, then by Proposition \ref{propintegral}, we obtain
\[
b_{\theta\theta}-\frac{\Gamma_\theta}{2\Theta}=-\beta^2b.
\]
or if we assume $\Gamma=\mbox{constant}$ we obtain.
\[
b_{\theta\theta}+\beta^2b=0.
\]
The last equation is the equation of an harmonic oscillator 
and its general solution is $b(\theta)=E_1\sin(\beta \theta)-E_2\cos(\beta\theta)$, where $E_1$ and $E_2$ are arbitrary constants. 
Thus we find a constant of motion $I$ of the following form 
\[
I=2aH+\Gamma p_{\theta}^2+2E_1I_1-2E_2I_2
\]
where $H$ is the Hamiltonian and 
\beq\begin{split}
&I_1=\sin(\beta \theta)p_r p_{\theta}-\beta\cos(\beta\theta)\Theta(r)p_{\theta}^2-\frac{\gamma}{\beta} \cos(\beta\theta),\\
&I_2=-\cos(\beta \theta)p_r p_{\theta}-\beta\sin(\beta\theta)\Theta(r)p_{\theta}^2-\frac{\gamma}{\beta} \sin(\beta\theta)
\end{split}
\eeq

%or 
%\[\begin{split}
%&I_1=J\sin[\beta(\theta-\theta_0)],\\
%&I_2=J\cos[\beta(\theta-\theta_0)]
%\end{split}
%\]
%where $p_rp_\theta=J\cos\beta\theta_0$ and $[\beta\Theta(r)+\frac{\gamma}{\beta}]=J\sin\beta\theta_0$ and $J=\sqrt{(p_rp_\theta)^2+[\beta\Theta(r)+\frac{\gamma}{\beta}]^2}$. 
Note that $2aH+\Gamma p_\theta^2$ is a constant of motion and, since $E_1$ and $E_2$ are arbitrary constants, 
$I_1$ and $I_2$ are integrals  of motions that generalize the Laplace-Runge-Lenz vector. 

The four conserved quantities $H,p_\theta,I_1$ and $I_2$ are related by the equation 
\beq
I_1^2+I_2^2=2p_\theta^2H-\frac{K^2p_\theta^4}{\beta^2}+\frac{\gamma^2}{\beta^2}
\label{eqdependent}
\eeq
giving three independent constants of motion.
Equation (\ref{eqdependent}) can be easily derived from
\[
I_1^2+I_2^2=p_r^2p_\theta^2+\beta^2\Theta^2p_\theta^4+2\gamma\Theta p_\theta^2+\frac{\gamma^2}{\beta^2}
\]
using Proposition \ref{proptheta2}, i.e. using that $\Theta^2=\frac{1}{\beta^2f^2}-\frac{K^2}{\beta^4}$. Substituting the  expression for $\Theta^2$ in the previous equation we obtain
\[
I_1^2+I_2^2=2p_\theta^2\left(\frac{p_r^2}{2}+\frac{p_\theta^2}{2f^2}+\gamma\Theta\right)-\frac{K^2p_\theta^4}{\beta^2}+\frac{\gamma^2}{\beta^2}.
\]

%%%%%%%%%%%%%%%%%
\section{Equation of the Trajectory}
%%%%%%%%%%%%%%%%
The shape and orientation of the orbits can be determined using the generalized Laplace-Runge-Lenz integrals:
\[
I_1\cos(\beta \theta)+I_2\sin(\beta\theta)=-\beta\Theta(r){p_{\theta}}^2-\frac{\gamma}{\beta}
\]
and thus
\[-\Theta(r)=\frac{\gamma}{\beta^2p_{\theta}^2}\left(1+\frac{\beta}{\gamma}(\cos(\beta\theta)I_1+\sin(\beta\theta)I_2\right)\]
or 
\beq
\rho=\frac{1}{p}(1+e\cos(\beta(\theta-\theta_0))
\eeq
where $\rho=-\Theta(r)$, $p=\left(\frac{\gamma}{\beta^2p_{\theta}^2}\right)^{-1}$, $(I_1\beta/\gamma)=e\cos(\beta\theta_0)$ and $(I_2\beta/\gamma)=e\sin(\beta\theta_0)$.

%%%%%%%%%%%%%%%%%%%%%%%
\section{A Geometric Description of The Flow}
\label{prova}
%%%%%%%%%%%%%%%%%%%%%%%%

From now on we consider only  the Kepler problem on spherical surface of revolution with positive constant  curvature.
The sphere is the trivial  example  of  spherical surface of revolution with positive constant Gaussian curvature. All the other surfaces of revolution of this kind have metric singularities at the north and south pole. An example  is depicted in figure \ref{surface}.

We now present a description of the orbit structure of the  system, with special emphasis on the orbits near collision. The coordinates used here are a generalizations of those  of R. McGehee used by several authors to study collisions in Newtonian gravitational systems \cite{McGehee}.
If the Gaussian curvature is constant and positive, by Proposition \ref{propintegral}, since $f(r)=L\sin(\sqrt{K}r+\alpha)$ and $\beta^2=L^2K$, we find
\[
\Theta(r)=-\frac{1}{L^2\sqrt{K}}\cot(\sqrt{K}r+\alpha).
\]
%in accordance with 
Clearly $f(r)=0$ when $\sin(\sqrt{K}r+\alpha)=0$, i.e. when $\sqrt{K}r+\alpha=0$
or when $\sqrt{K}r+\alpha=\pi$. Let $r_N=-\alpha/\sqrt{K}$ and $r_S=(\pi-\alpha)/\sqrt{K}$.
Then $f'(r)=L\sqrt{K}\cos(\sqrt{K}r+\alpha)$, thus $f'(r_N)=\sqrt{K}L>0$ and $f'(r_S)=-\sqrt{K}L<0$.

Consequently $V(r)$ has an attractive singularity at the north pole $r=r_N$ and an equal repulsive singularity at the south pole $r=r_S$.
Since $V:(r_N,r_S)\rightarrow \mathbb{R}$ is real analytic, standard results of differential equation theory guarantee, for any initial data $(r(0),\theta(0),p_r(0),p_\theta(0))$, the existence and uniqueness of an analytic solution defined on a maximal interval $[0,t^*)$, where $0<t^*\leq \infty$. If   $t^*<\infty$, we say the solution  is {\it singular}.
In general there are different kinds of singularities of the solutions. However in this problem we have only one kind: the singular solutions are such that  $r(t)\rightarrow r_N$  as $t \rightarrow t^*$. In this case  we say that the solution experience a {\it collision}.
Thus  the   singularity of the potential at  the north pole $r=r_N$   induces singularities in the solutions and  corresponds to a collision (on the other hand the singularity of the potential at the south pole $r=r_S$ does not induce singularities in the solutions).

To study the flow near collisions, i.e. near $r_N$  consider the following transformation of coordinates

\beq
\left \{
\begin{array}{l}
v=\frac{p_r}{\sqrt{|\Theta(r)}|}\\
u=\frac{p_\theta}{f(r)\sqrt{|\Theta(r)|}}.
\end{array} \right.
\eeq

In these coordinates, taking $m=1$, the original system  becomes,
\beq\left \{\begin{array}{l}
\frac{dr}{dt}=v\sqrt{|\Theta(r)|}\\
\frac{dv}{dt}=\left(u^2 f'(r)-\frac{v^2}{2f(r)\Theta(r)}-\frac{\gamma}{f(r)|\Theta(r)|}\right)\frac{\sqrt{|\Theta(r)|}}{f(r)}
\\
\frac{d\theta}{dt}=u\frac{\sqrt{|\Theta(r)|}}{f(r)}\\
 \frac{du}{dt}=-uv\left(f'(r)+\frac{1}{2f(r)\Theta(r)}\right)\frac{\sqrt{|\Theta(r)|}}{f(r)}
\end{array}\right.
\label{system1}
\eeq
If we take the energy integral $H$ to have constant value $h$ then  the energy relation gives
\beq
\sgn( \Theta(r))\left(\frac{u^2+v^2}{2}\right)+\gamma=\frac{h}{\Theta(r)}.
\label{eqenergy}
\eeq
Similarly if we take the angular momentum $p_\theta$ to have constant value $c$, then the angular momentum relation takes the form
\beq
u=\frac{c}{f(r)\sqrt{|\Theta(r)|}}.
\eeq
Moreover, since $p_rp_\theta=uvf(r)|\Theta(r)|$ the integrals $I_1$ and $I_2$ can be written as
\beq\begin{split}
&I_1=\sin(\beta \theta)|\Theta(r)|f(r)uv-\beta\cos(\beta\theta)\Theta(r)|\Theta(r)|f(r)^2u^2-\frac{\gamma}{\beta} \cos(\beta\theta),\\
&I_2=-\cos(\beta \theta)|\Theta(r)|f(r)uv-\beta\sin(\beta\theta)\Theta(r)|\Theta(r)|f(r)^2u^2-\frac{\gamma}{\beta} \sin(\beta\theta)
\end{split}
\label{eqlentz2}\eeq
%\[
%\frac{I_1}{|\Theta(r)|}=\bar J\sin[\beta(\theta-\theta_0)],\quad\quad
%\frac{I_2}{|\Theta(r)|}=\bar J\cos[\beta(\theta-\theta_0)]
%\]
%where $uvf(r)=\bar J\cos\beta\theta_0$, $[\beta~\sgn(\Theta(r))+\gamma/(\beta %|\Theta(r))|)]=\bar J\sin\beta\theta_0$ and  $J=|\Theta(r)|\bar J$.

The system (\ref{system1}) is no longer Hamiltonian, but (\ref{eqenergy}) defines a codimension one invariant set
 \beq
{\bf M}(h)=\{(r,\theta,u,v)\in \R^4|~r\in [r_N,r_S] \mbox{ and the energy relation %(\ref{eqenergy}) 
holds}\}.
\label{eqenergy1}
\eeq
which we continue to call energy  manifold. 

System (\ref{system1}) determines a vector field on ${\bf M}(h)$ which is undefined when $r=r_N$ or $r=r_S$. We now consider the singularity at $r=r_N$. We will come back to the singularity at $r=r_S$ in Section \ref{regularization}. Let
\beq\begin{split}
&{\bf N}=\{(r,\theta,u,v)\in {\bf M}(h)|~ r=r_N\}\\
\label{eqcollmanifold}
\end{split}
\eeq
which is the manifolds of states corresponding to collision. From the definition of ${\bf M}(h)$, we see that 
\[\begin{split}
&{\bf N}=\{(r,\theta,u,v)\in \R^4|~r=r_N \quad\mbox{and}\quad u^2+v^2=2\gamma \}\\
\end{split}
\]
and hence ${\bf N}$ is independent of $h$. 
Since $\theta$ is considered modulo $2\pi$, ${\bf N}$ is a two dimensional torus.

Note that the vector field given by (\ref{system1}) is not defined on ${\bf N}$. 
The orbits approaching ${\bf N}$  in a finite time are the collision orbits. However, we can scale the vector field in such a way that the new vector field can be extended to ${\bf N}$. We accomplish this
scaling by introducing a new time parameter  $\tau$ given by
\[
d\tau=\frac{\sqrt{|\Theta(r)|}}{f(r)}dt.
\]
System (\ref{eqmotion}) then, using Proposition \ref{propintegral}, becomes

\beq\left \{\begin{array}{l}
\dot r=vf(r)\\
\dot v=u^2 f'(r)-\frac{v^2}{2f(r)\Theta(r)}-\frac{\gamma}{f(r)|\Theta(r)|}\\
\dot \theta=u\\
\dot u=-uv\left(f'(r)+\frac{1}{2f(r)\Theta(r)}\right)
\end{array}\right.
\label{system2}
\eeq
where the dot indicates  differentiation with respect to $\tau$.
For this new vector field, ${\bf N}$ is an invariant set. We call ${\bf N}$ the {\itshape  collision manifold}. The flow on ${\bf N}$  is fictitious (i.e. it has no physical meaning), however, due to the continuity of solutions with respect to initial conditions, it gives informations about collision and near-collision solutions. 
%%%%%%%%%%%%%%%%%%%%%%%%%%%%%%%%%%%%%%%%%%%%%%%
\section{ Flow on the  Collision Manifold}
%%%%%%%%%%%%%%%%%%%%%%%%%%%%%%%%%%%%%%%%%%%%%%%%%
Let $k=\sqrt{K}L$. Since   $f'(r_N)=\sqrt{K}L>0$ where $\sqrt{K}L=|\beta|$ by Proposition \ref{propintegral}. Moreover, by  the same  proposition,  $f(r)\Theta(r)\rightarrow -k/\beta^2$ as $r\rightarrow 0$ and, near the north pole, we  have that $\Theta(r)<0$. 
It is easy to verify that in the variables $(r,v,\theta,u)$ the points $(r_N,\pm\sqrt{2\gamma},\theta,0)$ are equilibria for system (\ref{system2}). 
Therefore the sets
\beq
{\bf S}^{\pm}=\{(r,v,\theta,u)\in{\bf N}|v=\pm \sqrt{2\gamma}\}=\{(r,v,\theta,u)\in \R^4| (r,v,\theta,u)=(r_N,\pm\sqrt{2\gamma},\theta,0)\}
\eeq
 are circles of rest points on the collision manifold ${\bf N}$.
At these points  the linearized system has the matrix
\beq\left[
\begin{array}{cccc}
\pm k\sqrt{2\gamma}&0&0&0\\
*&\pm\beta^2\sqrt{2\gamma}/k&0&0\\
0&0&0&1\\
0&0&0&\mp\sqrt{2\gamma}(k-\beta^2/(2k))
\end{array}\right],
\label{eqeigenvalues}
\eeq
where $*$ denotes an entry of the matrix that is not used  in the computation of the eigenvalues. The corresponding eigenvalues are 
$\pm k\sqrt{2\gamma}$, $\pm\beta^2\sqrt{2\gamma}/k$, $0$ and $\mp\sqrt{2\gamma}(k-\beta^2/(2k))$.

Using   that $f(r)\Theta(r)\rightarrow -k/\beta^2$ as $r\rightarrow 0$ and  that $\Theta(r)<0$ near the north pole we see that the  restriction of system (\ref{system2}) to ${\bf N}$  yields the system 

\[\left \{\begin{array}{l}
\dot v=\left (k-\frac{\beta^2}{2k}\right)u^2\\
\dot \theta=u\\
\dot u=-uv\left(k-\frac{\beta^2}{2k}\right)
\end{array}\right.
\label{system3}
\]
where we have used  the energy relation to simplify $\dot v$.
Introducing the angular variable $\chi$ via
\beq
\begin{split}
u=\sqrt{2\gamma}\cos\chi\\
v=\sqrt{2\gamma}\sin\chi
\end{split}
\eeq
one finds
\beq
\frac{d\chi}{d\theta}=\left(k-\frac{\beta^2}{2k}\right)=\frac{|\beta|}{2}.
\eeq
The solutions of this vector field are sketched in figure \ref{figflow}.

%-------------------------------------------------------------------------------------
\begin{figure}[t]
  \begin{center}
      {\resizebox{!}{5.5cm}{\includegraphics{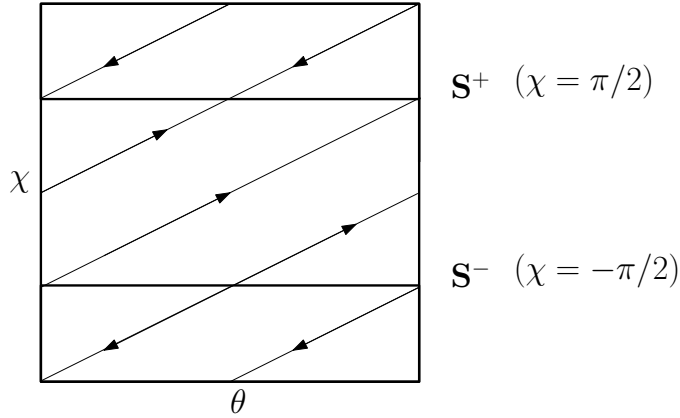}}  }
    \caption{The flow on the collision manifold ${\bf N}$ of the generalized Kepler problem\label{figflow} for $\beta=\frac 1 2$} 
  \end{center}

\end{figure}

In the resulting phase portrait $\dot v>0$
so all  orbits travel from the lower circle of rest points ${\bf S}^-$ to the upper circle ${\bf S}^+$. Moreover, in the  $(\theta,\chi)$ variables, the unstable manifolds of points  on ${\bf S}^-$ are just straight lines of slope $|\beta|/2$.
The  unstable manifolds at $\theta=\theta^*$, $v=-\sqrt{2\gamma}$
in ${\bf N}$ do not necessarily join up with the stable manifold at $\theta=\theta^*$, $v=\sqrt{2\gamma}$.
Only when we have 
\beq
\frac{d\chi}{d\theta}=\left(k-\frac{\beta^2}{2k}\right)=\frac{|\beta|}{2}=\frac{1}{2n} \quad \quad n=1,2,3,\ldots
\eeq
we have this property. That is if $|\beta|=1/n$ for a positive integer $n$, then  each branch of the unstable manifold makes $n$ circuits of ${\bf N}$ before rejoining the upper circle of rest points at the same $\theta$-value.

When 
\beq
\frac{d\chi}{d\theta}=\left(k-\frac{\beta^2}{2k}\right)=\frac{|\beta|}{2}=\frac{1}{2n+1} \quad \quad n=0,1,2,3,\ldots
\eeq
or equivalently when $|\beta|=2/(2n+1)$,
the unstable manifold leaving $\theta=\theta^*$ join up with the stable manifold at $\theta=\theta^*+\pi$ after making $n+1/2$ circuits. 

In all other cases, the two branches of the unstable manifold reach distinct equilibrium points. 

{\it Physical Interpretation.} We now give a physical interpretation of the solution described above. As we mentioned earlier, the orbits on the collision manifold have no physical meaning, but give information about collision and near collision orbits.
A near collision orbit makes $n$ revolutions about $r=0$ before exiting at an angle which depends on $\beta$.
In the two special cases considered above the orbit either exits in the direction  in which approached collision ($|\beta|=1/n$) or else in exactly the opposite direction ($|\beta|=2/(2n+1)$).

If $\beta$ is not of one of these two forms, then nearby  initial conditions will lead to quite different behavior near collision (this is the basic idea behind Easton's notion of topological regularization). In this case we cannot join orbits coming to collision with orbits leaving collision in a meaningful way so to make the resulting flow  continuous. 

%%%%%%%%%%%%%%%%%%%%%%%%%%%%%%%%%%%%%%%%%%%%%%
\section{An Isolating Block About Collisions}
%%%%%%%%%%%%%%%%%%%%%%%%%%%%%%%%%%%%%%%%%%%%%%

Let ${\bf M}$ be a smooth manifold and let $\psi:{\bf M}\times \R\rightarrow {\bf M}$ be a flow on ${\bf M}$. A subset ${\bf N}\subset {\bf M}$ is called invariant if $\psi({\bf N},\R)={\bf N}$.

\begin{definition}
 A compact invariant set ${\bf N}\subset {\bf M}$ is called  {\itshape isolated} if there exist an open set ${\bf U}$ containing ${\N}$ such that $\psi(x,\R)\in{\bf U}$ implies $x\in {\bf N}$. The set ${\bf U}$ is called isolating neighborhood for ${\bf N}$.
\end{definition}

Now let ${\bf B}$ be a compact subset of ${\bf M}$ with non-empty interior and suppose that ${\bf b}=\pd{\bf B}$ is a smooth submanifold of ${\bf M}$. 
Denote by ${\bf b}^+$ the set of {\itshape ingress points} of ${\bf b}$,

\[{\bf b}^+=\{x\in{\bf b}|~\psi(x,(0,-\epsilon))\cap {\bf B}=\emptyset \mbox{ for some } \epsilon>0)\},\]
by ${\bf b}^-$ the set of {\itshape egress points} of ${\bf b}$,
\[{\bf b}^-=\{x\in{\bf b}|~\psi(x,(0,\epsilon))\cap {\bf B}=\emptyset \mbox{ for some } \epsilon>0)\},\]
and by ${\bf t}$ the set of tangency points of ${\bf b}$
\[{\bf t}=\{x\in{\bf b}|~ \dot \psi(x,0) \mbox{ is tangent to } {\bf b} \mbox{ at } x\}.\]
In general ${\bf b}^+,{\bf b}^-$ and ${\bf t}$ might be variously related but their union must always be ${\bf b}$: points of ${\bf b}^+$ leave ${\bf B}$ going backwards, those in ${\bf b}^-$ leave ${\bf b}$ going forwards, and the remainder must be in ${\bf t}$.

\begin{definition}
${\bf B}$ is called an {\itshape isolating block} if ${\bf t}={\bf b}^+\cap{\bf b}^-$ (and ${\bf t}$ is a smooth submanifold of ${\bf b}$ with codimension 1).
\end{definition}

\begin{definition}
 Let ${\bf N}$ be an isolated invariant set, and let ${\bf B}$ be an isolating block. Then ${\bf B}$ is said to {\itshape isolate} ${\bf N}$ if $\mbox{int}({\bf B})$ is an isolating neighborhood for ${\bf N}$
\end{definition}

The following theorem was proved by Conley and Easton \cite{ConleyEaston}

\begin{theorem}
If ${\bf N}$ is an isolated invariant set, then there exists an isolating block which isolates ${\bf N}$. If ${\bf B}$ is an isolating block, then there exists an isolated invariant set (possibly empty) which is isolated by ${\bf B}$. 
\end{theorem}

We now want to define ${\bf B}$ in terms of a real valued function on ${\bf M}$. Let $I:{\bf M}\rightarrow \R$ be a smooth map.  We write 
\[
I^*(x,t):=I(\psi(x,t)),
\]
and define
\[
\dot I(x)=\dot I^*(x,0)\quad \mbox{and}\quad \ddot I(x)=\ddot I^*(x,0)\]
where $\dot I^*$ and $\ddot I^*$ denote derivatives with respect to time. The following lemma is proved by Wilson and Yorke \cite{Wilson} (the symbol ``D'' denotes derivative). 

\begin{lemma}
Let $I:{\bf M}\rightarrow[0,\infty)$, and let $\delta_0>0$. Suppose that $DI(x)\neq 0$ whenever $0<I(x)\leq \delta_0$. Suppose also that $\ddot I(x)>0$ whenever $0<I(x)<\delta_0$ and $\dot I(x)=0$. Then ${\bf N}\equiv I^{-1}(0)$ is an isolated invariant set and $I^{-1}([0,\delta])$ is an isolating block for ${\bf N}$ for each $\delta\in(0,\delta_0].$
\label{Wilson}
\end{lemma}
We now define the subsets of ${\bf b}$ that are asymptotic to ${\bf N}$:
\[\begin{split}
 &  {\bf a}^+=\{x\in {\bf b}^+|\psi(x,[0,\infty))\subset{\bf B}\},\\
 &{\bf a}^-=\{x\in {\bf b}^-|\psi(x,(-\infty,0])\subset{\bf B}\}.\\
  \end{split}
\]
By definition  $ {\bf b}^+-{\bf a}^+$ denotes the set of point in ${\bf b}$ with the property that the corresponding solutions enter in ${\bf B}$ without staying permanently there, i.e.  there exists a $t>0$ such that $\psi(x,t)\not\in {\bf B} $.
 Thus we can define the time spent in the block for a point $x\in{\bf b}^+-{\bf a}^+$ by
\[
T(x)=\inf\{t>0|\psi(x,t)\not\in{\bf B}\}.
\]
Note that $\psi(x,[0,T(x)])\in {\bf B}$ and that $\psi(x,T(x))\in {\bf b}^-$. Now we define the map across the block
\[
\Psi:{\bf b}^+-{\bf a}^+\rightarrow {\bf b}^--{\bf a}^-:x\rightarrow \psi(x,T(x)).
\]

\begin{theorem}[Conley and Easton \cite{ConleyEaston}]
 If ${\bf B}$ is an isolating block, then 
$\Psi:{\bf b}^+-{\bf a}^+\rightarrow {\bf b}^--{\bf a}^-$ is a diffeomorphism
\label{thConleyEaston}
\end{theorem}

\begin{definition}
An isolating block {\bf B} is said to be {\it trivializable} 
if $ \Psi$  extends uniquely to a diffeomorphism from ${\bf b}^+$ to ${\bf b}^-$.
\end{definition}

Trivializability is actually a property of an isolated invariant set: 
\begin{lemma}[Conley\cite{Conley}] Suppose that ${\bf N}$ is an isolated invariant set and that ${\bf B}_1$ and ${\bf B}_2$ isolate ${\bf N}$. Then ${\bf B}_1$ is trivializable if and only if ${\bf B}_2$ is trivializable. 
\end{lemma}
\begin{definition}
 Let ${\bf B}$ isolate ${\bf N}$. Then {\bf N} is said to be {\it trivializable} if ${\it B}$ is trivializable. 
\end{definition}

%%%%%%%%%%%%%%%%%%%
We now return to the generalized Kepler problems on spherical surfaces of revolutions. We take $\psi$ to be the flow on the manifold ${\bf M}(h)$ defined by equations (\ref{eqenergy1}). The invariant set ${\bf N}$ is given by (\ref{eqcollmanifold}).  
Define 
\[\begin{split}
&I:\M(h)\rightarrow \R:(r,v,\theta,u)\rightarrow \frac{1}{|\Theta(r)|} \\
&\B(h,\delta)=\{x\in \M(h):I(\bf x)\leq \delta\}
\end{split}
\]

\begin{lemma}
Given any $h$, there exists a $\delta_0>0$ such that $\B(h,\delta)$ is an isolating block for ${\bf N}$  whenever $0<\delta\leq \delta_0$.
\label{lemmaMcGehee}
\end{lemma}
\begin{proof}
The tangent space to $\M(h)$ at the point $x=(r,v,\theta,u)$ is given by
\[\{\frac{h}{f(r)^2\Theta(r)^2}\dot r+ v\dot v+u\dot u\}\]
provided $(\frac{h}{f(r)^2\Theta(r)^2},v,0,u)\neq(0,0,0,0)$. Let $\dot x=(\dot r,\dot v,\dot \theta,\dot u)$. Since $DI(x)\cdot \dot x=\frac{\dot r}{f(r)^2\Theta(r)^2}$, we have that $DI(x)\neq 0$ if $(u,v)\neq(0,0)$.  The energy relation implies that 
\[
|\sgn(\Theta(r))(u^2+v^2)+2\gamma|=2|h|\delta
\]
where $I(\delta)=\delta>0$. Therefore if $\delta_0$ is small enough (i.e. $\delta_0<\gamma/|h|$), then $(u,v)\neq(0,0)$ whenever $\delta\leq \delta_0$. Hence $DI(x)\neq 0$ when $0<I(x)\leq \delta_0$. Now using equation (\ref{system2}) we see that
\[
\dot I=-\frac{\sgn(\Theta(r))}{f(r)\Theta(r)^2}v=\frac{1}{f(r)\Theta(r)^2}v\]
where the last equality follows from the fact that near the north pole $\Theta(r)<0$.
Moreover (again assuming that $\Theta(r)<0$, since we are near the north pole)
 \[\ddot I=\frac{1}{f(r)\Theta(r)^2}\left [u^2f'(r)-v^2\left(\frac{5/2}{f(r)\Theta(r)}+f'(r)\right)+\frac{\gamma}{f(r)\Theta(r)}\right].
\]
If $I(x)=\delta$ and if $\dot I(x)=0$, then $|\theta(r)|^{-1}=\delta$ and $v=0$. Using the energy relation and Proposition \ref{propintegral} yields,
\[\begin{split}
\ddot I &=\frac{2}{f(r)\Theta(r)^2}\left(f'(r)h\delta+\gamma\left(f'(r)+\frac{1}{2f(r)\Theta(r)}\right)\right)\\
&=\frac{2}{f(r)\Theta(r)^2f'(r)}\left[(h\delta+\gamma)f'(r)^2-\gamma\frac{\beta^2}{2}\right].
\end{split}
\]
Since $0<\delta\leq\delta_0<\gamma/|h|$ we have that $h\delta+\gamma>0$ and the expression in square brakets is a convex parabola in $f'(r)$.
Moreover, since  $f'(r_N)=|\beta|$, it follows that,  if $\delta_0$ is sufficientely small, 
then $\left[(h\delta+\gamma)f'(r)^2-\gamma\frac{\beta^2}{2}\right]>0$ and $f'(r)>0$.
Consequently $\ddot I> 0$ whenever $0<\delta\leq\delta_0$.  Finally, note that ${\bf N}=I^{-1}(0)$. Hence, by Lemma \ref{Wilson}, ${\bf B}(h,\delta)$ is an isolating block for ${\bf N}$. 
\end{proof}

We now exhibit for the block ${\bf B}(h,\delta)$  the various subsets defined above.
We fix $h$ and choose $0<\delta\leq \delta_0$, where $\delta_0$ is given in Lemma  \ref{lemmaMcGehee}.
\[\begin{split}
 &{\bf b}=\{x\in {\bf M}(h)|~ \Theta(r)=\delta\}\\
  & {\bf b}^+=\{x\in {\bf b}|~ v\leq 0\}\\
 &{\bf b}^-=\{x\in {\bf b}|~ v\geq 0\}\\
 &{\bf t}=\{x\in {\bf b}|~ v=0\}\\
 &{\bf a}^+=\{x\in {\bf b}^+|~ u=0\}\\
 &{\bf a}^-=\{x\in {\bf b}^-|~ u=0\}\\
  \end{split}
\]

%%%%%           THEOREM 3

\begin{theorem}
 If the set ${\bf N}$ is a trivializable isolated invariant set for equations (\ref{system2}) then  $|\beta|=2/m$, where $m$ is a positive integer.
%The set ${\bf N}$ is a trivializable isolated invariant set for equations (\ref{system2})   %if and only if $|\beta|=2/m$, where $m$ is a positive integer.
\end{theorem}
\begin{proof}
Using the definition of ${\bf B}$, we write
\[
{\bf b}=\{(r,v,\theta,u)\in \R^4| \frac{1}{|\Theta(r)|}=\delta \mbox{ and } u^2+v^2=2\gamma+2h\delta\}.
\]
%In order to show that the singularity at the north pole  is regularizable  we must prove that the map across the block
The general form of the map across the block is 
\[
\Psi:{\bf b}^+-{\bf a}^+\rightarrow {\bf b}^--{\bf a}^-:(r,v,\theta,u)\rightarrow(\bar r,\bar v,\bar\theta,\bar u)
\]
we will use the first integral of the system to simplify this expression. 
%extends to a diffeomorphism  of ${\bf b}^+$ onto ${\bf b}^-$.
Suppose $(r,v,\theta,u)\in {\bf b}^+-{\bf a}^+$ and let $\Psi(r,v,\theta,u)=(\bar r,\bar v,\bar\theta,\bar u)$. 
Recall that the energy $h(r,v,\theta,u)$, the angular momentum $c(r,v,\theta,u)$, $I_1(r,v,\theta,u)$ and $I_2(r,v,\theta,u)$ are first integrals of the problem.
Then $h(r,v,\theta,u)=h(\bar r,\bar v,\bar\theta,\bar u)$, $I_1(r,v,\theta,u)=I_1(\bar r,\bar v,\bar\theta,\bar u)$, $I_2(r,v,\theta,u)=I_2(\bar r,\bar v,\bar\theta,\bar u)$ and $c(r,v,\theta,u)=c(\bar r,\bar v,\bar\theta,\bar u)$. 
Since $r$ is constant in ${\bf b}$ it follows that $\bar r=r$ and  the integral of the angular momentum yields $\bar u=u$. From the energy integral and the definition of ${\bf b}^+$ and ${\bf b}^-$ it follows that $\bar v=-v$.
Therefore we can write the map across the block as
\beq
\Psi:{\bf b}^+-{\bf a}^+\rightarrow {\bf b}^--{\bf a}^-:(r,v,\theta,u)\rightarrow(r,-v,\Psi_\theta(\theta,u),u)
\eeq
where $\Psi_\theta(\theta,u)$ is the third component of $\Psi$.
Here we are using $(\theta,u)$  as  coordinates on ${\bf b}^+$, so $1/|\Theta(r)|=\delta$
and $v=-(2\gamma+2h\delta-u^2)^{1/2}$. 
Since the equations (\ref{system2}) are independent of $\theta$
\[
\Psi_\theta(\theta,u)=\theta+\Gamma(u).
\]
The function $\Gamma $ is defined for all $u$ such that $0<u^2\leq 2\gamma +2h\delta $. By symmetry,
\beq
\Gamma(-u)=-\Gamma(u).
\label{eqsymm}
\eeq
Assume that ${\bf B}$ is trivializable. Then $\Psi$ extends to a continuous map ${\bf b}^+\rightarrow {\bf b}^-$. Thus
\[
\theta+\Gamma(0+)=\theta +\Gamma(0-)+2\pi m.
\]
where $m$ is an integer. By (\ref{eqsymm}), $\Gamma(0-)=-\Gamma(0+)$. Hence we must have
\[
\Gamma(0+)=\pi m.
\]

The number $\Gamma(0+)$ can be computed using geometric methods.
Consider a point $x_0\in {\bf a}^+$.  The orbit through $x_0$ is the stable manifold of a point ${\bf s}^-$ in ${\bf S}^-$. Now let $x\in {\bf b}^+$ be close to $x_0$. the orbit through $x$ follow closely the stable manifold of ${\bf s}^-$.

In Section 7 we studied the flow on ${\bf N}$ where we found out that the unstable manifolds of points on ${\bf S}^-$ are just stright lines with slope $|\beta|/2$. 
We are interested in the branch of the unstable manifold of ${\bf s}^-$ for which $u\geq0$. Therefore we take $-\pi/2\leq\chi\leq \pi/2$.
Write
\[
{\bf s}^-=(\bar\delta,\theta^*,0,-\sqrt{2\gamma}),
\]
where $1/|\Theta(\bar\delta)=\delta|$. Then the unstable manifold of ${\bf s}^-$ is exactly the stable manifold of the point
\[
{\bf s}^+=(\bar\delta,\theta^*+2\pi/|\beta|,0,\sqrt{2\gamma}).
\]

We now can determine $\Gamma(0+)$. The orbit through $x$ first follows the stable manifold of ${\bf s}^-$, then follows the unstable manifold of ${\bf s}^-$, which coincides with the stable manifold of ${\bf s}^+$, and finally follows the unstable manifold of ${\bf s}^+$.
Therefore, as $x\rightarrow x_0$, the change in $\theta$ along the orbit approaches the difference in $\theta$ between ${\bf s}^-$ and ${\bf s^+}$. This difference is $2\pi/|\beta|$.
Hence $\Gamma(0+)=2\pi/|\beta|$. Since we also have that $\Gamma(0+)=\pi m$ we obtain
$|\beta|=2/m$. Since $|\beta|\geq 0$, $m$ is positive. This completes the proof.  
\end{proof}

We now want to prove that ${\bf N}$ is a trivializable isolated invariant set for equations (\ref{system2}). In order to do that, unlike  in the work of McGehee \cite{McGehee}, we will use the additional integrals of motion of the system. 

%%%%             THEOREM 4

\begin{theorem}
 If  $|\beta|=2/m$, where $m$ is a positive integer, then  the set ${\bf N}$ is a trivializable isolated invariant set for equations (\ref{system2}).
\end{theorem}
\begin{proof}
Recall that 
\[
{\bf b}=\{(r,v,\theta,u)\in \R^4| \frac{1}{|\Theta(r)|}=\delta \mbox{ and } u^2+v^2=2\gamma+2h\delta\}.
\]
 
Suppose $(r,v,\theta,u)\in {\bf b}^+-{\bf a}^+$ and let $\Psi(r,v,\theta,u)=(\bar r,\bar v,\bar\theta,\bar u)$. 
%Recall that the energy $h(r,v,\theta,u)$, the angular momentum $c(r,v,\theta,u)$, $I_1(r,v,\theta,u)$ and $I_2(r,v,\theta,u)$ are first integrals of the problem.
%Then $h(r,v,\theta,u)=h(\bar r,\bar v,\bar\theta,\bar u)$, $I_1(r,v,\theta,u)=I_1(\bar r,\bar v,\bar\theta,\bar u)$, $I_2(r,v,\theta,u)=I_2(\bar r,\bar v,\bar\theta,\bar u)$ and $c(r,v,\theta,u)=c(\bar r,\bar v,\bar\theta,\bar u)$. 
%Since $r$ is constant in ${\bf b}$ it follows that $\bar r=r$ and  the integral of the angular momentum yields $\bar u=u$. From the energy integral and the definition of ${\bf b}^+$ and ${\bf b}^-$ it follows that $\bar v=-v$.
Then, in order to show that ${\bf N}$ is trivializable, we must prove that the map across the block
\beq
\Psi:{\bf b}^+-{\bf a}^+\rightarrow {\bf b}^--{\bf a}^-:(r,v,\theta,u)\rightarrow(\bar r,\bar v,\bar\theta,\bar u)
\eeq
extends to a diffeomorphism  of ${\bf b}^+$ onto ${\bf b}^-$.
Here, as shown in the proof of Theorem 3, $\Psi_\theta(\theta,u)$ is the third component of $\Psi$, and  we are using $(\theta,u)$  as  coordinates on ${\bf b}^+$, so $1/|\Theta(r)|=\delta$
and $v=-(2\gamma+2h\delta-u^2)^{1/2}$.

Since $I_1$ and $I_2$ are integrals of motions  we have
$I_1(r,v,\theta,u)=I_1(\bar r,\bar v,\bar\theta,\bar u)$ and, $I_2(r,v,\theta,u)=I_2(\bar r,\bar v,\bar\theta,\bar u)$. Therefore we have 
\[
J_1(r,v,\theta,u)=I_1\cos \beta\theta+I_2\sin\beta\theta=I\cos(\beta(\theta-\theta_0))=-\beta\Theta(r)|\Theta(r)|f(r)^2 u^2-\frac\gamma \beta
\]
and 
\[
J_2(r,v,\theta,u)=I_1\sin \beta\theta-I_2\cos\beta\theta=I\sin(\beta(\theta-\theta_0))=|\Theta(r)|f(r)uv,
\]
where 
\beq
I_1=I\cos\beta\theta_0, \quad\quad I_2=I\sin\beta\theta_0
\label{eqtheta0}\eeq and $I=\sqrt{I_1^2+I_2^2}$.
Consequently $J_1(r,v,\theta,u)=J_1(\bar r,\bar v,\bar\theta,\bar u)$ and $J_2(r,v,\theta,u)=-J_2(\bar r,\bar v,\bar\theta,\bar u)$. This gives the system of equations
\[
\begin{split}
 &I\cos(\beta(\theta-\theta_0))=I\cos(\beta(\bar\theta-\bar\theta_0))\\
&I\sin(\beta(\theta-\theta_0))=-I\sin(\beta(\bar\theta-\bar\theta_0))
\end{split}
\]
where $\bar\theta_0=\theta_0(\bar r,\bar v,\bar\theta,\bar u)$.
Since $\theta_0=\bar\theta_0$ the general solution of the system above is
\beq
\Psi_\theta=\bar\theta=-\theta+2\theta_0+\frac{2k\pi}{\beta}.
\label{eqpsitheta}
\eeq
We now want to write  $\Psi_\theta(\theta,u)$ in the form $\Psi_\theta(\theta,u)=\theta+\Gamma(u)$ used in the proof of Theorem 3.
First we can rewrite equations (\ref{eqlentz2}) as 
\beq
I_1= I\sin [\beta(\theta-\zeta)], \quad I_2=-I\cos[\beta(\theta-\zeta)],
\label{eqlentz3}\eeq
where $\cos(\beta\zeta)=(|\Theta(r)|f(r)uv)/I$ and $\sin(\beta\zeta)=(\beta\Theta(r)|\Theta(r)|f(r)^2u^2+\gamma/\beta)/I$. If we suppose 
 $(r,v,\theta,u)\in {\bf b}^+-{\bf a}^+$ then $1/|\Theta(r)|=\delta$ and if we use $(\theta,u)$ as coordinates on ${\bf b}^+$ we have
\beq\begin{split}
&\cos(\beta\zeta)=\frac{-k_1u(2\gamma+2h\delta-u^2)^{1/2}}{(k_1^2(2\gamma+2h\delta-u^2)u^2+(k_2u^2+\frac{\gamma}{\beta})^2)^{1/2}},\\ 
&\sin(\beta\zeta)=\frac{(k_2u^2+\frac{\gamma}{\beta})}{(k_1^2(2\gamma+2h\delta-u^2)u^2+(k_2u^2+\frac{\gamma}{\beta})^2)^{1/2}},
\end{split}\label{eqzeta}\eeq
where $k_1$ and $k_2$ are constants such that $k_1=|\Theta(r)|f(r)$ and $k_2=\beta\Theta(r)|\Theta(r)|f(r)^2u^2+\gamma/\beta$, when $1/|\Theta(r)|=\delta$.

Note that the right hand side of the second equation of system (\ref{eqzeta}) is always positive. It follows that $0<\beta\zeta<\pi$. Therefore the function $\zeta$ is defined by
\[
\zeta(u)=\frac 1 \beta \arccos\left( \frac{-k_1u(2\gamma+2h\delta-u^2)^{1/2}}{(k_1^2(2\gamma+2h\delta-u^2)u^2+(k_2u^2+\frac{\gamma}{\beta})^2)^{1/2}}\right)
\]
and it is of class $C^\infty$ since composition of $C^\infty$ functions (the inverse cosine function is of class $C^\infty$). 

From equations (\ref{eqtheta0}) and (\ref{eqlentz3}) we obtain  the following system of  trigonometric equations 
\[
\cos\beta\theta_0=\sin [\beta(\theta-\zeta)], \quad\quad\\
\sin\beta\theta_0=-\cos[\beta(\theta-\zeta)].\]

The general solution of the system above is 
$\theta_0=\theta-\zeta -\frac{\pi}{2\beta}+\frac{2n\pi}{\beta}$. We choose $n=0$ and we define $\theta_0$ as
\[
\theta_0=\theta-\zeta -\frac{\pi}{2\beta}.
\]
Substituting the above expression for $\theta_0$ in (\ref{eqpsitheta}) we obtain
\[
\Psi_\theta=\theta+\Gamma(u)=\theta-2\zeta-\frac{\pi}{\beta}+\frac{2k\pi}{\beta}
\]
where $\Gamma(u)=-2\zeta-\frac{\pi}{\beta}+\frac{2k\pi}{\beta}$, and $k$ is an integer to be determined.

We now need to determine the value of $k$ imposing continuity at $u=0$.
Recall that $\Gamma(0+)=\pi m=\frac{2\pi}{|\beta|}$. Since $\beta\zeta(0+)=\frac \pi 2$ we have
$\Gamma(0+)=\frac{2\pi}{\beta}(k-1)=\frac{2\pi}{|\beta|}$. Consequently, if $u>0$ 
$k=1+\frac{\beta}{|\beta|}$. Since $\Gamma(0-)=-\Gamma(0+)$ and $\Gamma(0-)=\frac{2\pi}{\beta}(k-1)=\frac{2\pi}{|\beta|}$ we find that, if $u<0$, $k=1-\frac{\beta}{|\beta|}$. 
Therefore we can define $\Gamma(u)$ as
\[
\Gamma(u)=\left \{
\begin{array}{l}
-2\zeta-\frac{\pi}{\beta}+\frac{2\pi}{\beta}\left(1-\frac{\beta}{|\beta|}\right)\quad \mbox{if } u<0 \\
-2\zeta-\frac{\pi}{\beta}+\frac{2\pi}{\beta}\left(1+\frac{\beta}{|\beta|}\right) \quad \mbox{if } u>0.
\end{array}\right.
\]

We can now show that the map across the block $\Psi$ extends  to a diffeomorphism $\tilde \Psi$ from ${\bf b}^+$ to ${\bf b}^-$. Let 
\[
\tilde\Psi:{\bf b}^+\rightarrow {\bf b}^-:(r,v,\theta,u)\rightarrow(r,-v,\theta+\tilde\Gamma(u),u),
\]
where 
\[
\tilde\Gamma(u)=\left \{
\begin{array}{l}
-2\zeta-\frac{\pi}{\beta}+\frac{2\pi}{\beta}\left(1-\frac{\beta}{|\beta|}\right)\quad \mbox{if } u\leq 0 \\
-2\zeta-\frac{\pi}{\beta}+\frac{2\pi}{\beta}\left(1+\frac{\beta}{|\beta|}\right) \quad \mbox{if } u>0
\end{array}\right.
\]
be the extended map. Such map is continuous by construction. Moreover, since $\zeta(u)$ is differentiable, $\tilde\Psi$ is a differentiable map $\tilde\Psi$. Moreover, let $(\bar r,\bar v,\bar\theta,\bar u)\in {\bf b}^-$, then $\tilde\Psi^{-1}:{\bf b}^-\rightarrow{\bf b}^+$ is defined as follows
\[
\tilde\Psi^{-1}(\bar r,\bar v,\bar\theta,\bar u)=(\bar r, \bar\theta-\tilde\Gamma(\bar u),\bar u,-\bar v)
\] 
is the inverse of the map  $\tilde\Psi$. The differentiability of $\tilde \Psi^{-1}$ follow immediately. Therefore $\Psi$ extends to a diffeomorphism, this concludes the proof. 
\end{proof}
%Now we use $I_1, I_2$ and the angular momentum to determine $\Psi_\theta$.
%We must have $I_1(r,v,\theta,u)=I_1(r,-v,\Psi_\theta(\theta,u),u)$ and $I_2(r,v,\theta,u)=I_2(r,-v,\Psi_\theta(\theta,u),u)$. Consequentely we obtain 
%\[
%J\sin[\beta(\theta-\theta_0)]=J\sin[\beta(\Psi_\theta-\theta_0)]
%\]

%%%%%%%%%%%%%%%%%%%%%%%%%%%%%%%%%%%%%%%%%%%%%%%%%%%%%%
\section{Block Regularization of the Vector Field}\label{regularization}
%%%%%%%%%%%%%%%%%%%%%%%%%%%%%%%%%%%%%%%%%%%%%%%%%%%%%%%%%
In the previous section we worked with system (\ref{system2}), for which ${\bf N}$ is an invariant set. 
 We now turn to our original problem: determining if equations (\ref{eqmotion}) can be regularized.  We consider the equivalent system (\ref{system1}). The set ${\bf N}$ given by (\ref{eqcollmanifold}) is the set of singularities of equations (\ref{system1}), i.e. it is the set where the vector field fails to be defined.  On ${\bf M}(h)-{\bf N}$ the orbits for the two sysyems are identical: only the parametrization is different. 
We now turn, following Easton, to the definition of block regularization. 

Let ${\bf M}$ be a smooth manifold, let ${\bf N}$ be a compact subset of ${\bf M}$, and let ${\bf F}$ be a vector field on ${\bf M-N}$, where, in  this section, ${\bf N}$ is the set of singularities of the vector field ${\bf F}$. Let $\phi$ be the flow on ${\bf M-N}$ given by ${\bf F}$ (we do not require  $\phi(x,t)$ to be defined for all $t$). 

Let ${\bf B}$ be a compact subset of ${\bf M}$ with non-empty interior, and suppose that ${b}=\partial{\bf  B}$ is a smooth manifold  which does not intersect ${\bf N}$. Let ${\bf b^+}, {\bf b^-}$ and ${\bf t}$  be defined as in the previous section. Let  the definition of isolating block be also the same. Let ${\mathcal O}(x)$ denote the orbit through $x$, namely
\[
{\mathcal O}(x)=\{\phi(x,t)|\phi(x,t) \mbox{ is defined}\}.
\] 

\begin{definition}
 An isolating block ${\bf B}$ is said to isolate a singularity set ${\bf N}$ if ${\bf N}\subset\mbox{int}({\bf B})$ and if ${\mathcal O}(x)\not\subset{\bf B}$ for all $x\in {\bf B-N}$.
\end{definition}
 The subsets ${\bf a}^+$ and ${\bf a}^-$ are the same as before, except that now we must allow for solutions which are not defined for all $t$. Therefore
\[
\begin{split}
 &  {\bf a}^+=\{x\in {\bf b}^+|\phi(x,t)\in{\bf B} \mbox{ for all } t\geq 0 \mbox{ for which } \phi(x,t) \mbox{ is defined}\},\\
 &{\bf a}^-=\{x\in {\bf b}^-|\phi(x,t)\in{\bf B} \mbox{ for all } t\leq 0 \mbox{ for which } \phi(x,t) \mbox{ is defined}\}.\\
  \end{split}
\]
The map $\Phi:{\bf b}^+-{\bf a}^+\rightarrow {\bf b}^-$ is defined in exactly the same way as the map $\Psi$ and we have that $\Phi:{\bf b}^+-{\bf a}^+\rightarrow {\bf b}^--{\bf a}^-$ is a diffeomorphism. We also have the same definition of a trivializable block ${\bf B}$. 
\begin{definition}
 The singularity set ${\bf N}$ is said to be {\it block regularizable} if there exists a trivializable block ${\bf B}$ which isolates ${\bf N}$.
\end{definition}

We now return to our original problem. Note that whether a certain set is or is not an isolating block is independent of the parametrization of the flow. The map across the block is also independent of the parametrization.  Therefore, ${\bf B}(h,\delta)$ is an isolating block for system (\ref{system1}) if and only if it is an isolating block for system (\ref{system2}), and $\Phi=\Psi$. Hence ${\bf B}(h,\delta)$ is trivializable for (\ref{system1}) if and only if it is trivializable for system (\ref{system2}).

\begin{theorem}
 The singularity set ${\bf N}$ (i.e. the singularity at the north pole) for system (\ref{system1}) is block regularizable if and only if $|\beta|=2/m$, where $m$ is a positive integer. 
\end{theorem}

We can now quickly discuss the singularity set at the south pole for system (\ref{system1}). In this case the situation is much simpler and the singularity is always block regularizable.
 Let ${\bf N}_S$ be such singularity set, i.e. 
\[
{\bf N}_S=\{(r,\theta,u,v)\in \R^4|r=r_S\}
\]
then one can prove the following
\begin{theorem}
 The singularity set ${\bf N}_S$ for system (\ref{system1}) is block regularizable.
\end{theorem}
\begin{proof}
We give a sketch of the proof. One can construct an isolating block repeating the proof  of Lemma \ref{lemmaMcGehee} with some minor changes. In this case the sets ${\bf a}^+$ and ${\bf a^-}$ are empty. This is a consequence of the fact that, as it can be seen from the energy relation (\ref{eqenergy}), no solution ever reaches the singularity set ${\bf N}_S$.
It follows from Theorem \ref{thConleyEaston} that the map across the block $\Psi:{\bf b^+}-{\bf a}^+\rightarrow {\bf b^-}-{\bf a}^-$ is a diffeomorphism. Moreover, since  ${\bf a}^+$ and ${\bf a^-}$ are empty, $\Psi$ trivially extends to a diffeomorphism   from ${\bf b^+}$ to ${\bf b^-}$. Therefore  the isolating block is trivializable.  
Moreover, as we noticed above, a region is an isolating block for system (\ref{system1}) if and only if it is an isolating block for (\ref{system2}), and $\Phi=\Psi$.
Since the isolating block can be constructed so that it isolates ${\bf N}_S$ it follows that ${\bf N}_S$ is block regularizable. 
  \end{proof}
%%%%%%%%%%%%%%%%%%%%%%%%%%%%%%%%%%%%%%%%%%%%%%%%%%%%
\section{A Note About Orbifolds of Revolutions}
%%%%%%%%%%%%%%%%%%%%%%%%%%%%%%%%%%%%%%%%%%%%%%%%%%

%-------------------------------------------------------------------------------------
\begin{figure}[t]
  \begin{center}
      {\resizebox{!}{4.5cm}{\includegraphics{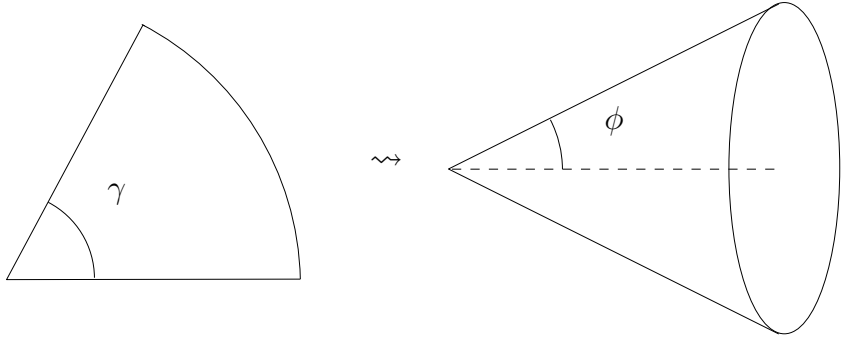}}  }
    \caption{Cone as quotient of a plane sector \label{figcone}} 
  \end{center}

\end{figure}
%--------------------------------------------------------------------------
\iffalse
\begin{figure}[t]%[thb]
  \begin{center}
\scalebox{0.6} % Change this value to rescale the drawing.
{
\begin{pspicture}(0,-0.02)(23.2,5.97)
\psline[linewidth=0.04cm](13.2,3.05)(19.2,5.95)
\psline[linewidth=0.04cm](13.2,3.05)(19.2,0.05)
\psellipse[linewidth=0.04,dimen=outer](19.2,3.0)(1.0,2.95)
\pswedge[linewidth=0.04](3.15,0.0){6.15}{0.0}{65.556046}
\psarc[linewidth=0.04](3.05,0.0){2.15}{0.0}{64.29005}
\usefont{T1}{ptm}{m}{n}
\rput(6.028281,1.295){\Large $\gamma$}
\usefont{T1}{ptm}{m}{n}
\rput(10.918282,3.095){\Large $\leadsto$}
\psline[linewidth=0.04cm,linestyle=dotted,dotsep=0.16cm](13.2,3.05)(19.2,3.05)
\psarc[linewidth=0.04](12.95,3.0){2.25}{0.0}{23.749495}
\usefont{T1}{ptm}{m}{n}
\rput(16.04828,3.595){\Large $\phi$}
\end{pspicture} 
}
\caption{Cone as quotient of a plane sector \label{figcone}}
  \end{center}
\end{figure}
\fi
Spherical orbifolds of revolutions are easily identifiable by their tangent cones at the poles. Namely, the tangent cone at the pole must be isometric to the metric quotient of the flat plane $\R^2$ by a finite cyclic group of rotations fixing the origin. Note that the tangent cone at the pole is generated by rotating the tangent line to the profile curve at the pole about the axis of rotation. If the cyclic groups at the poles are of different orders, the orbifold is commonly referred as {\itshape bad} since it will not arise as a quotient of a Riemannian $S^2$ by a finite cyclic group of isometries \cite{Thurston}.

In general a flat right circular cone with vertex angle $\phi$ is obtained by identifying the edges of a plane circular sector of angle $\gamma$. The relation between $\phi$ and $\gamma$ is easily computed: $\gamma=2\pi\sin\phi$. See figure \ref{figcone}. Thus if the tangent cone at a pole of a spherical orbifold of revolution is isometric to $\R^2/{\mathbb Z}_n$, then $\gamma=2\pi/n$ for a positive integer $n$.

We are now in a position to prove the following.
\begin{theorem}
 The singularity set ${\bf N}$ for system (\ref{system1}) is block regularizable if the surface of revolution is an orbifold of revolution with constant curvature. 
\end{theorem}
\begin{proof}
Let $\phi_N$, resp. $\phi_S$, be the angle between the profile curve $\gamma(r)=(f(r),0,g(r))$ and the axis of rotation at $r=r_N$, resp. $r=r_S$. 
For an orbifold of revolution we must have $\sin\phi_N=1/n$ and $\sin\phi_S=1/k$ for some positive integers $n$ and $k$. 

On the other hand $\sin\phi_N=f'(r_N)=|\beta|$. By Theorem 5 the singularity ${\bf N}$ for system (\ref{system1})  is block regularizable if and only if $|\beta|=2/m$, where $m$ is a positive integer.  

Therefore for an orbifold of revolution we have $m=2n$. A similar argument also shows $m=2k$. This proves that the singularity ${\bf N}$ is always block  regularizable on orbifolds of revolution with constant Gaussian curvature.   
\end{proof}

\section*{Acknowledgments}
The author acknowledges with gratitude useful discussions pertinent to the present research with Florin Diacu, Ray McLenaghan, B. Doug Park, Ernesto Perez-Chavela, and Cristina Stoica. The research was supported in part by an NSERC Discovery grant.

\end{document}